\documentclass[prl, aps, twocolumn]{revtex4-2}

\usepackage[utf8]{inputenc}
\usepackage{graphicx}
\usepackage{tabularx}
\usepackage{amsmath}
\usepackage{amssymb}
\usepackage[dvipsnames]{xcolor}
\usepackage[colorlinks=true,
            linkcolor=red,
            citecolor=red,
            urlcolor=red]{hyperref}

\usepackage{siunitx}
\usepackage{float}
\usepackage{amsthm}
\newtheorem*{theorem}{Theorem}
\newtheorem*{lemma}{Lemma}

\begin{document}

\title{Entanglement Complexity in Many-body Systems from Positivity Scaling Laws}

\author{Anna O. Schouten}
\author{David A. Mazziotti}
\email{damazz@uchicago.edu}
\affiliation{Department of Chemistry and The James Franck Institute, The University of Chicago, Chicago, IL 60637 USA}

\date{Submitted September 2, 2025}

\begin{abstract}
Area laws describe how entanglement entropy scales and thus provide important necessary conditions for efficient quantum many-body simulation, but they do not, by themselves, yield a direct measure of computational complexity. Here we introduce a complementary framework based on $p$-particle positivity conditions from reduced density matrix (RDM) theory. These conditions form a hierarchy of $N$-representability constraints for an RDM to correspond to a valid $N$-particle quantum system, becoming exact when the Hamiltonian can be expressed as a convex combination of positive semidefinite $p$-particle operators. We prove a general complexity bound: if a quantum system is solvable with level-$p$ positivity independent of its size, then its entanglement complexity scales polynomially with order $p$. This theorem connects structural constraints on RDMs with computational tractability and provides a rigorous framework for certifying when many-body methods including RDM methods can efficiently simulate correlated quantum matter and materials.
\end{abstract}

\maketitle

{\em Introduction:} Entanglement entropy scaling and area laws~\cite{Horodecki2009, Amico2008, Eisert2010, Kais.2007} are important concepts for characterizing entanglement in a number of disciplines including quantum information~\cite{Kitaev2006, Levin2006}, black hole physics~\cite{Bombelli1986, Srednicki1993}, and quantum many-body physics~\cite{Hastings2004, Plenio2005, Hastings2007a, Schollwoeck2005}.  While one might expect the entropy to scale with the volume of the subsystem---an extensive property---, in a system that obeys an area law the entanglement entropy scales with the surface area of the subsystem~\cite{Amico2008, Eisert2010}. In many-body physics, area laws have important implications both for qualifying the entanglement complexity of the system, i.e., interpreting the nature of quantum correlations, and assessing the computational tractability of obtaining numerical solutions~\cite{Hastings2007, Schuch2008}. For example, whether a system obeys an area law is an indication of its Hamiltonian's locality, and the extent to which its correlations can be represented in a size-independent framework~\cite{Eisert2010}. Significantly, this implies that a system that obeys an area law is solvable at polynomial cost, which is seen in the context of density matrix renormalization group (DMRG) method~\cite{White1992, Schollwoeck2005}, where systems that obey area laws can be efficiently described using matrix-product states~\cite{Verstraete2006}.

Here we introduce an alternative measure of entanglement complexity based on $p$-positivity from the perspective of reduced density matrix (RDM) theory~\cite{Mazziotti2012, Mazziotti2023}. The $p$-positivity conditions, which place constraints on the RDMs at the level of $p$ particles, form a hierarchy of $N$-representability conditions that are necessary to ensure that a computed RDM accurately represents an $N$-particle density matrix~\cite{Coleman1963, Kummer1967, Coleman2000, Mazziotti2012, Mazziotti2023}.  The level of $p$-positivity required to obtain a solution to a many-body problem provides a measure of the entanglement complexity of the system.  To demonstrate, we prove a general complexity bound: if a quantum system is solvable at a fixed level of $p$ positivity conditions that is independent of its overall size, then its entanglement and solution complexity are bounded by a polynomial of order $p$.  Not only is the $N$-representability of the $p$- and fewer-body RDMs achievable with no more than $p$-body operators, but also by duality, the Hamiltonian is expressible as a convex combination of no more than $p$-particle positive semidefinite operators~\cite{Mazziotti2012, Mazziotti2023, Mazziotti2001}.  This structure is closely analogous to the sum-of-squares framework in Lasserre’s hierarchy for nonnegative polynomials, where convex combinations of positive semidefinite forms enforce positivity~\cite{Lasserre2001}  We illustrate this result for the extended Hubbard model which, with $t = 0$, is exactly solvable at the level of 2-positivity.  The positivity scaling laws provide a rigorous framework for not only quantifying entanglement complexity but also certifying when RDM theory, such as the variational 2-RDM method~\cite{Mazziotti2004, Mazziotti2007, Mazziotti2012, Nakata2001, Zhao2004, Shenvi2010, Verstichel.2011, Knight.2022, DePrinceIII2024}, one-electron RDM methods~\cite{Piris.2021, Schilling.2021}, and the related RDM bootstrapping approaches~\cite{Han2020, Gao2025, Scheer.2024}, can efficiently simulate correlated quantum matter and materials. \\

{\em Theory:} For a Hamiltonian, $\hat{H}$, representing an electronic many-body system, the interactions are in general at most pairwise~\cite{Coleman1963}. Consequently, the energy can be written as a function of the 2-particle reduced density matrix (2-RDM), $^2D$,
\begin{equation}
    E = \mathrm{Tr}(\hat{H}\ ^2D).
    \label{eq:energy}
\end{equation}
For the 2-RDM to correctly represent a valid $N$-particle density matrix, beyond the requirements for a density matrix---Hermiticity, fixed trace, particle-exchange symmetry, and positive semidefiniteness---,the 2-RDM must be subject to additional constraints, known as $N$-representability conditions~\cite{Coleman1963, Kummer1967, Coleman2000, Mazziotti2012, Mazziotti2023}.  A hierarchy of such constraints, known as the $p$-positivity conditions~\cite{Mazziotti2001, Mazziotti2012, Mazziotti2006}, can be defined as constraints on the $p$-RDM ($^{p} D$)
\begin{equation}
    \mathrm{Tr}(\hat{C}_i^{}\hat{C}_i^{\dagger}\ ^{p}D)\geq 0, \forall\ \hat{C}^{}_i
    \label{eq:p-positivity}
\end{equation}
where $\hat{C}_i$ are polynomials of order $p$ in fermionic creation and annihilation operators, $\hat{a}^{\dagger}$ and $\hat{a}$. Eq.~(\ref{eq:p-positivity}) constrain all metric matrices associated with the $p$-body RDM to be positive semidefinite~\cite{Mazziotti2001}, meaning all eigenvalues of the matrix must be nonnegative.  The $p$-positivity conditions can be enforced on a lower $q$-RDM with $q<p$ such as the 2-RDM in one of two ways: ({\em i}) contracting the constrained $p$-RDM to the $q$-RDM~\cite{Mazziotti2001, Mazziotti2006} or ({\em ii}) taking convex combinations of the constraints to generate constraints directly on the $q$-RDM---the $(q,p)$-positivity conditions~\cite{Mazziotti2012, Mazziotti2023}.  This method of constrained optimization as functional of the 2-RDM is known as variational 2-RDM (V2RDM) theory~\cite{Mazziotti2004, Mazziotti2007, Mazziotti2012, Nakata2001, Zhao2004, Shenvi2010, Verstichel.2011, Knight.2022, DePrinceIII2024}, which has been applied to computing the ground-state energies and properties of strongly correlated molecules and materials~\cite{Xie.20222n, Schouten.2023, Torres.2024yzf, Schouten.2025a9e}.  The strategy for obtaining an iterative solution subject to physically or symmetrically derived constraints can also be referred to as bootstrapping~\cite{Ferrara1973, Polyakov1974, Poland2019}. For the problem solved at the level of $p$-positivity, the minimized energy $E^{(p)}$ approaches the exact $N$-body energy $E^{(N)}$ from below, making $E^{(p)}$ a lower bound on the exact energy.

At the solution to Eq.~(\ref{eq:energy}), the following must be satisfied,
\begin{gather}
    \mathrm{Tr}[(\hat{H} - E)\ ^pD] = 0 \label{eq:dual} \\
    \mathrm{Tr}[(^p\hat{C}_{i}^{}{} ^p\hat{C}_{i}^{\dagger})\ ^pD] = 0, \forall\ ^pD.
    \label{eq:dual2}
\end{gather}
Because the operators $^p\hat{C}_{i}^{}{} ^p\hat{C}_{i}^{\dagger}$ in Eq.~(\ref{eq:dual2}), defining the active $N$-representability conditions of $^pD$ for the given Hamiltonian~\cite{Mazziotti2012}, span the null space of $^pD$, the Hamiltonian $(\hat{H} - E)$ in Eq.~(\ref{eq:dual}) can be written as a convex combination of these operators such that,
\begin{equation}
    (\hat{H} - E) = \sum_i \beta_i\ ^p\hat{C}_{i}^{}\ ^p\hat{C}_i^{\dagger}
    \label{eq:Hamiltonian_sum}
\end{equation}
where $\beta_i \geq 0$.  If we maximize the energy $E$ subject to the constraints on the Hamiltonian in Eq.~({\ref{eq:Hamiltonian_sum}), then we have the dual-cone (or polar-cone) formulation, developed and implemented in Refs.~\cite{Mazziotti2023, Mazziotti.2020z0p, Mazziotti.2016sve, Mazziotti2012, Cances.2006}.  While RDMs do not appear explicitly in the dual formulation, we have previously shown that they arise indirectly where the Lagrange multiplier of the $p$-particle constraints on the energy maximization in Eq.~(\ref{eq:Hamiltonian_sum}) is the $p$-RDM~\cite{Mazziotti.2020z0p}.

Ref.~\cite{Mazziotti2012} proves that the $N$-representability conditions on the 2-RDM can be expressed as a hierarchy---$(2,p)$-positivity conditions---in which 2-body operators are constructed from convex combinations of positivity semidefinite $p$-body operators and that when $p=r$ in the hierarchy, where $r$ is the rank of the one-electron basis, the $N$-representability conditions are complete.  From this perspective, expressing $(\hat{H} - E)$ as a convex combination of positive semidefinite operators generates one of the $N$-representability conditions on the 2-RDM.

The dual-cone variational 2-RDM theory enables the exploitation of Hamiltonian structure including sparsity, low rank behavior, and symmetries.  It has been applied to molecular systems with the 2-positivity conditions as well as the partial 3-positivity conditions including the $T_{2}$ conditions~\cite{ Mazziotti.2020z0p, Mazziotti.2016sve}.  Recently, a formulation in the dual cone has been developed that solves a matrix equation to generate all of the constraints on the 2-RDM, providing another approach to the $N$-representability conditions~\cite{Mazziotti2023} that complements the solution by convex combinations in Ref.~\cite{Mazziotti2012}.

While the preceding discussion has focused on the ground-state problem, the $p$-positivity formulation is not inherently limited to the ground state. By modifying the Hamiltonian in the objective---for example, by using a variance Hamiltonian $({\hat H}-E)^2$~\cite{Mazziotti2001}---the same framework can in principle be applied to other quantum states. For clarity of exposition, however, we restrict the statements below to the ground-state problem, noting that the results for entanglement complexity extend more broadly.

\begin{lemma}
    If a problem can be expressed exactly as an energy minimization subject to the $p$-positivity conditions in Eq.~(\ref{eq:p-positivity}) or equivalently, an energy maximization constrained by the convex combination of positive semidefinite $p$-particle operators as in Eq.~(\ref{eq:Hamiltonian_sum}), then the solution can be obtained in polynomial time at fixed $p$.
\end{lemma}

\begin{proof}
    The optimization of the energy with the $p$-positivity constraints in Eq.~(\ref{eq:p-positivity}) or Eq.~(\ref{eq:Hamiltonian_sum}) may be solved using a semidefinite program~\cite{Mazziotti2004, Mazziotti2011, Fukuda2007}. Since semidefinite programs can be solved in polynomial time~\cite{Vandenberghe.1996}, the solution of a quantum problem exactly expressible at fixed $p$ is polynomial in cost.
\end{proof}

\noindent The lemma shows that whenever a problem is expressible at a fixed level of $p$-positivity, its solution can be obtained with polynomial cost; the crucial issue, however, is whether that level can remain fixed as the system grows.

\begin{theorem}
If a problem is solvable at the level of $p$-positivity in a manner that is independent of system size---i.e., a single $p$ suffices uniformly in $N$ and, when defined, in the thermodynamic limit---, then both the entanglement complexity and the solution complexity are $O(p)$.
\end{theorem}

\begin{proof}
If the level $p$ needed to represent the Hamiltonian does not increase with system size, then there exists a solution at level-$p$ positivity with $p<N$. By the lemma, such a problem can be solved with polynomial complexity, so the solution complexity is $O(p)$. Furthermore, because the quantum problem is representable with only the $p$-RDM in the primal formulation and convex combinations of $p$-body operators in the dual, the correlations---and thus the entanglement---are supported only up to $p$ particles, $p$ holes, or any combination of particles and holes up to $p$. Hence, the entanglement complexity is also $O(p)$.
\end{proof}

The above theorem establishes a significant class of quantum problems with polynomial-scaling entanglement complexity that are solvable in the context of reduced density matrix theory by semidefinite programming in polynomial time.  Examples where 2-positivity is sufficient include the antisymmetrized geminal power (AGP) model of superconductivity~\cite{Coleman2000} and the harmonic-interaction model for bosons~\cite{Gidofalvi.2004p89}.  When the complexity of the Hamiltonian approaches $N$, however, the level of positivity required to solve the problem also increases with $N$, and hence, the cost to solve such problems exhibits exponential scaling.  Such entanglement complexity can occur in highly correlated quantum scenarios such as in the vicinity of the critical points of some quantum phase transitions. \\

\begin{figure}[tbh!]
    \centering
    \includegraphics[width=8cm]{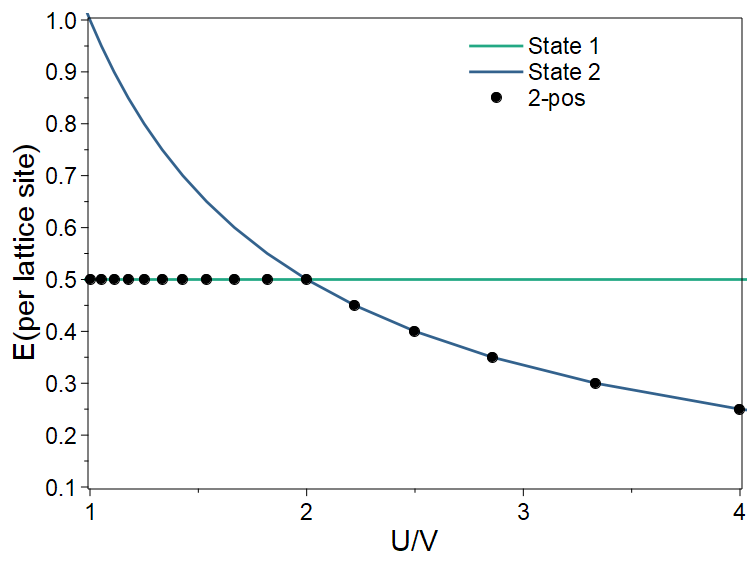}
    \caption{Energy of the extended Hubbard model with $t = 0$ and $U = 1$. In the region where $U/V < 2$, the ground state energy is equal to $U/2$. A phase transition occurs at $U/V = 2$, after which the ground state energy is equal to $V$. The two energy curves corresponding the $U$ (State 1) and $V$ (State 2) driven states are shown as solid lines. The ground state energy calculated with V2RDM with 2-positivity conditions (2-pos) is exact for this model and is shown as black points along the curves.}
    \label{fig:energy_t0}
\end{figure}

{\em Results:} We illustrate the theorem by solving the extended Hubbard model~\cite{Zhang1989, Qin2022, Arovas2022}, given by,
\begin{equation}
    \begin{split}
    \hat{H} &= -t\sum_{<i,j>} \hat{a}^{\dagger}_{i}\hat{a}^{}_j + h.c\\
    &+ U\sum_i \hat{a}^{\dagger}_{i\uparrow}\hat{a}^{}_{i\uparrow}\hat{a}^{\dagger}_{i\downarrow}\hat{a}^{}_{i\downarrow} + V \sum_{\substack{<i,j>\\ \sigma,\sigma'}} \hat{a}^{\dagger}_{i\sigma} \hat{a}^{}_{i\sigma}\hat{a}^{\dagger}_{j\sigma'}\hat{a}^{}_{j\sigma'}
\end{split}
\label{eq:extended_Hubbard}
\end{equation}
where $\hat{a} (\hat{a}^{\dagger})$ are fermionic creation and annihilation operators, $<i,j>$ indicates $i$ and $j$ are nearest-neighbors, $t$ represents hopping, $U$ represents on-site repulsion, and $V$ represents nearest-neighbor repulsion. We solve the model at half-filling with periodic boundary conditions, computing the ground-state energy from a variational minimization of the 2-RDM subject to 2-positivity or 2- and partial 3-positivity conditions~\cite{Mazziotti2004, Mazziotti2007, Mazziotti2012, Nakata2001, Zhao2004, Shenvi2010, Verstichel.2011, Knight.2022, DePrinceIII2024}.  The semidefinite program in the variational 2-RDM (V2RDM) calculation is solved by a first-order boundary-point algorithm~\cite{Mazziotti2011}.

In the model when $t = 0$ (i.e., no hopping), the Hamiltonian is diagonal and consists exclusively of two-body interactions. In this regime, the Hamiltonian has exact solutions with a phase transition from a charge-density wave (CDW) to a spin-density wave (SDW), accompanied by a discontinuous transition in the ground-state energy as the system moves between $U$-dominated and $V$-dominated regions~\cite{Hirsch1984, Dongen1992, Jeckelmann2002}. The ground-state energy in each region is given as a function of $U$ and $V$,
\begin{equation}
   E =  \begin{cases}
        UL/2 & U < 2V\\
        VL & U > 2V
    \end{cases}
    \label{eq:phase_transition}
\end{equation}
for $L$ lattice sites. In this regime, the Hamiltonian is exactly solvable at the level of 2-positivity using semidefinite programming. The 2-positivity conditions constrain the two-particle ($^2D$), two-hole ($^2Q$), particle-hole ($^2G$) RDMs~\cite{Mazziotti2007, Garrod.1964},
\begin{gather}
    ^2D^{i,j}_{k,l} = \langle\Psi|\hat{a}^{\dagger}_i\hat{a}^{\dagger}_j\hat{a}^{}_l\hat{a}^{}_k|\Psi\rangle\\
    ^2Q^{i,j}_{k,l} = \langle\Psi|\hat{a}^{}_i\hat{a}^{}_j\hat{a}^{\dagger}_l \hat{a}^{\dagger}_k|\Psi\rangle\\
    ^2G^{i,j}_{k,l} = \langle\Psi|\hat{a}^{\dagger}_i\hat{a}^{}_j\hat{a}^{\dagger}_l\hat{a}^{}_k|\Psi\rangle,
\end{gather}
to be positive semidefinite.

Figure~\ref{fig:energy_t0} shows the exact energy at the two competing limits as a function of $U/V$ when $t = 0$ and $U = 1$ for a range of $V$. Because $U$ is fixed, the energy of the state dominated by $U$ is constant for all values of $U/V$. The state dominated by $V$ begins higher in energy when $U/V < 2$ and the two states cross at $U/V$ such that the $V$ dominated state becomes the ground state when $U/V > 2$. The ground-state energy calculated at the level of 2-positivity (2-pos) captures the transition of the ground-state from $E/L = U/2$ to $E/L = V$. Even at the phase transition, the complexity is of order 2 in this case, and the system can be solved exactly at the level of 2-positivity.

\begin{figure}[tbh!]
    \centering
    \includegraphics[width=8cm]{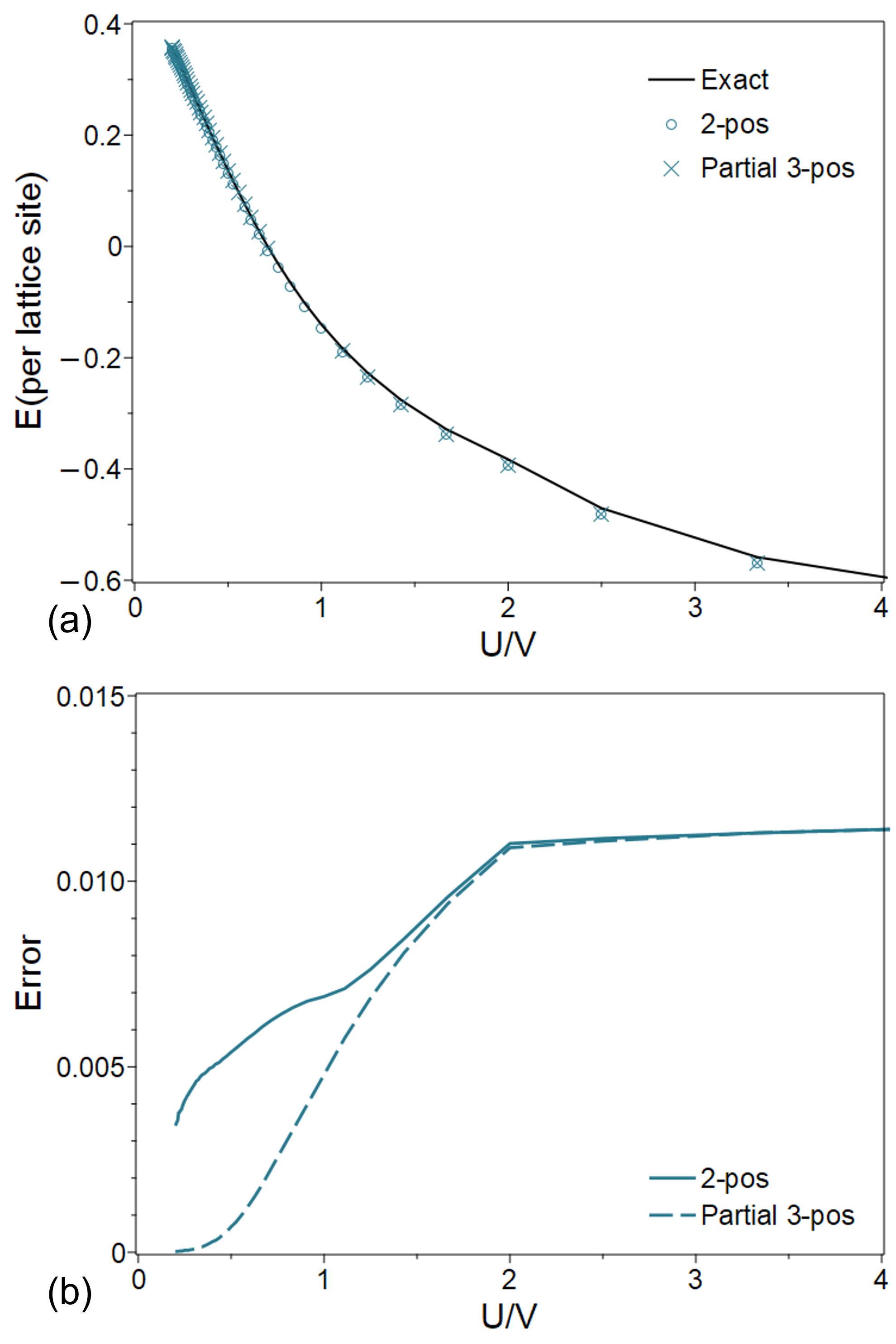}
    \caption{(a) The exact ground-state energy of the extended Hubbard model with $t = U = 1$ and the ground-state energy calculated with V2RDM with 2-positivity (2-pos) and partial 3-positivity (partial 3-pos) conditions. (b) The error between the exact ground-state energy and the V2RDM energy.}
    \label{fig:energy_error_t1}
\end{figure}

When $t > 0$, while there is still predicted to be phase transition near $U/V \approx 2$, the transition becomes continuous and 2-positivity conditions are not exact in either the $U > 2V$ or $U < 2V$ limits. Figure~\ref{fig:energy_error_t1}a shows the energy calculated using exact diagonalization and the semidefinite program with 2-positivity and partial 3-positivity conditions. A partial 3-positivity condition, known as the $T_2$ condition~\cite{Zhao2004, Mazziotti2006, Mazziotti.2016sve, Erdahl.1978}, imposes the following constraint in addition to 2-positivity,
\begin{gather}
    T_2 =\ ^3E\ +\ ^3F \succeq 0.
\end{gather}
where $^3E$ and $^3F$ are 3-body RDMs, representing the probabilities for two particles and a hole as well as two holes and a particle, respectively,
\begin{gather}
    ^3E^{i,j,k}_{l,m,n} = \langle\Psi|\hat{a}^{\dagger}_i\hat{a}^{\dagger}_j\hat{a}^{}_k\hat{a}^{\dagger}_n\hat{a}^{}_m\hat{a}^{}_l|\Psi\rangle\\
    ^3F^{i,j,k}_{l,m,n} = \langle\Psi|\hat{a}^{\dagger}_i\hat{a}^{}_j\hat{a}^{}_k\hat{a}^{\dagger}_n\hat{a}^{\dagger}_m\hat{a}^{}_l|\Psi\rangle.
\end{gather}
The $T_2$ matrix can be explicitly expressed in terms of the elements of the 2-RDM~\cite{Zhao2004, Mazziotti2006}. This condition is known as a partial 3-positivity condition and imposing this constraint adds a higher level of positivity than 2-positivity without complete 3-positivity.

The energies calculated with 2-positivity and partial 3-positivity conditions closely follow the shape of the curve for the exact energy. Unlike when $t = 0$, there are no discontinuities in the energy curve indicating the phase transition. However, examining the error of the variational results relative to the exact results, we observe a sharp change in the trend of the error around $U/V = 2$. When $U/V > 2$, the error curves calculated with 2-positivity conditions and partial 3-positivity conditions are almost the same and the curves are flat as $U/V$ increases. When $U/V < 2$, the error curves drop abruptly and the curve with 2-positivity conditions separates from the curve with partial 3-positivity conditions as the error with partial 3-positivity conditions approaches zero for small values of $U/V$. This behavior demonstrates a change in the complexity of the Hamiltonian around $U/V = 2$ corresponding to the phase transition. Above $U/V  = 2$ the complexity is relatively constant but the order of complexity exceeds 2-positivity and partial 3-positivity. Below $U/V = 2$, the complexity of the Hamiltonian decreases with $U/V$ and at small values of $U/V$ approaches a complexity order near the level of partial 3-positivity. \\

{\em Conclusions:} We present a theorem that connects entanglement complexity to $p$-positivity. Specifically, if a system is solvable at level-$p$ positivity in a manner independent of its size, then both the entanglement and solution complexities are $O(p)$. This result establishes a reduced-density-matrix perspective: a quantum problem solvable at level $p$ is representable through the $p$-RDM and its dual cone, which capture correlations involving at most $p$ particles or holes. As discussed in the accompanying lemma, such problems can be solved with semidefinite programming in polynomial time.  The extended Hubbard model illustrates this perspective, being exactly solvable at 2-positivity when $t=0$ and well-approximated by finite $p$ levels when $t>0$.

For certain types of systems near criticality the complexity increases exponentially, as indicated by violation of area laws~\cite{Amico2008}. While we do not make direct comparisons here, we expect that in these cases where area laws are violated due to increasing complexity, exponential growth in the required level of $p$-positivity would also occur. For example, in the Ising model at the phase transition, the complexity increases such that a level of positivity less than $N$ is insufficient to capture the rapid change in the 2-RDM near the critical point~\cite{Schwerdtfeger2009}. Nevertheless, in many cases, even if the complexity exceeds $p$, a finite level of $p$-positivity can offer reasonable solutions. This behavior occurs when the complexity is approximately reducible to finite $p$ even if the reduction is not exact, as seen in molecular calculations~\cite{Xie.20222n, Schouten.2023, Torres.2024yzf, Schouten.2025a9e} from amorphous coordination polymers to exciton condensates,  ultracold few-fermion systems~\cite{Knight.2022}, fractional quantum Hall states~\cite{Gao2025},  as well as spin systems like the extended Hubbard model with $t > 0$ discussed here.

The $p$-positivity framework thus complements the concept of area laws: both provide measures of entanglement complexity but through different structural lenses. In addition to offering a rigorous theory for solution complexity, $p$-positivity establishes a fundamental conceptualization of entanglement complexity within the context of reduced density matrices with potentially significant implications for our understanding of many-body quantum systems.

\begin{acknowledgments}
D.A.M gratefully acknowledges the U.S. National Science Foundation Grant No. CHE-2155082 for support.
\end{acknowledgments}

\bibliography{2-pos}

\end{document}